\documentclass[letterpaper,11pt]{article}

\author{Ali Mashreghi \footnote{Department of Computer Science, University of Victoria, BC, Canada. Email: ali.mashreghi87@gmail.com} \and Valerie King \footnote{Department of Computer Science, University of Victoria, BC, Canada. Email: val@uvic.ca, Funded with an NSERC grant.} }

\usepackage{url}
\usepackage{graphicx}
\usepackage{framed}
\usepackage{amssymb}
\usepackage{amsmath}
\usepackage{algorithm, algpseudocode}
\newcommand{\old}[1]{}

\newtheorem{theorem}{Theorem}[section]
\newtheorem{lemma}[theorem]{Lemma}
\newtheorem{observation}[theorem]{Observation}

\newenvironment{proof}{{\bf Proof:\ }}{\hfill$\Box$\medskip}

\date{}

\begin{document}
\title{Faster asynchronous MST and low diameter tree construction with sublinear communication} 
\maketitle

\begin{abstract}
Building a spanning tree, minimum spanning tree (MST), and BFS tree in a distributed network are fundamental problems which are still not fully understood in terms of time and communication cost. x The first work to succeed in computing a spanning tree with communication sublinear in the number of edges in an asynchronous CONGEST network appeared in DISC 2018. That algorithm which constructs an MST is sequential in the worst case; its running time is proportional to the total number of messages sent.  Our paper matches its message complexity but brings the running time down to linear in $n$.  Our techniques can also be used to provide an asynchronous algorithm with sublinear communication to construct a tree in which the distance from a source to each node is within an additive term of $\sqrt{n}$ of its actual distance. 

We can convert any asynchronous MST algorithm with time $T(n, m)$ and message complexity of $M(n, m)$ to an algorithm with time $O(n^{1 - 2\epsilon} + T(n, n^{3/2 + \epsilon}))$ and message complexity of $\tilde{O}(n^{3/2 + \epsilon} + M(n, n^{3/2+\epsilon}))$, for $\epsilon \in [0, 1/4]$. Picking $\epsilon = 0$ and using Awerbuch's algorithm \cite{awerbuch1987optimal}, this results in an MST algorithm with time $O(n)$ and message complexity $\tilde{O}(n^{3/2})$. 
However, if there were an asynchronous MST algorithm that takes time sublinear in $n$ and requires messages linear in $m$, by picking $\epsilon > 0$ we could achieve sublinear time (in $n$) and sublinear communication (in $m$), simultaneously. To the best of our knowledge, there is no such algorithm.

All the algorithms presented here are Monte Carlo and succeed with high probability, in the KT1 CONGEST asynchronous model.  

\end{abstract}

\newpage 

\section{Introduction}
A distributed network of processes can be represented as an undirected graph $G=(V, E)$, where $|V| = n$ and $|E| = m$. Each node corresponds to a  process and each edge corresponds to a communication link between the two processes. The nodes can communicate only by passing messages to each other. Computing the spanning tree and the minimum spanning tree (MST) are problems of fundamental importance in distributed computing.\old{ A spanning tree allows for a source to broadcast with $n-1$ messages. In a weighted graph where the weights of the edges reflect the cost of sending messages from one node to another, the MST is the most cost-effective network that can be used for communication. }
Efficient solutions for building these trees directly improve the solution to other distributed computing problems or at least provide valuable insights. Leader election, counting, and shortest path tree are examples of such problems. The breadth-first search tree (BFS)  is also important; it can be used to simulate a synchronous algorithm in an asynchronous network. 

The problem of constructing an MST in a distributed network has been studied for many years. In the earlier works, researchers focused on improving the time complexity since it was believed that any spanning tree algorithm in the CONGEST model would require $\Omega(m)$ messages (See \cite{awerbuch1990trade}). \old{The recent algorithms for computing the MST have time complexity close to the proved lower bounds.} After the algorithm of King et al. \cite{king2015construction}  which constructs the MST in the synchronous CONGEST model in $\tilde{O}(n)$ time and messages, there has been renewed interest in message complexity. 
Mashreghi and King \cite{mashreghi_disc_2018} achieved the first algorithm to compute a spanning tree in an asynchronous CONGEST model with $o(m)$ communication complexity when $m$ is sufficiently large. However,  the time complexity of their algorithm matches its communication complexity,  $\tilde{O}(n^{3/2})$; in the worst case, their algorithm essentially operates in a sequential manner. 
 
In the present work, we 
 match the message complexity of \cite{mashreghi_disc_2018} but bring the running time down to $O(n)$, a time which matches the time of the fastest known asynchronous MST algorithms that use $\Theta(m)$ communication  \cite{awerbuch1987optimal}.
 
 The classic Layered BFS algorithm for the asynchronous CONGEST network uses $O(D^2)$ time and $O(Dn+m)$ messages, where $D$ is the diameter of the network. We show how to construct a ``nearly BFS"  tree  in this model with sublinear number of messages (for small enough $D$, and sufficiently large $m$) such that for each node, the distance from the source node is within an additive term of $O(\sqrt{n})$ from its actual distance in the network. Such a tree can be used to simulate a synchronous algorithm in an asynchronous network with an overhead of $O(D+\sqrt{n})$ time per step. To the best of our knowledge, there is no previously known algorithm to construct a low diameter tree using a sublinear number of messages in an asynchronous network.

\old{
 Table \ref{tableA} gives a summary of the MST algorithms that use $o(m)$ communication for $m$ sufficiently small. All algorithms in this table are Monte Carlo and work in the KT1 CONGEST model.
\begin{table}[ht]
\caption{MST algorithms with sublinear communication in the number of edges.} 
\label{tableA}
\centering 
\begin{tabular}{c c c c} 
\hline\hline 
 Authors & Synchrony & Time Complexity & Message Complexity \\ [0.5ex] 
\hline 
King, Kutten, and Thorup \cite{king2015construction} & Synchronous & $\tilde{O}(n)$ & $\tilde{O}(n)$ \\
Mashreghi and King \cite{mashreghi2017time} & Synchronous & $O(\frac{n}{\epsilon})$ & $\tilde{O}(\frac{n^{1+\epsilon}}{\epsilon})$ \\
Mashreghi and King \cite{mashreghi_disc_2018} & Asynchronous & $\tilde{O}(n^{3/2})$ & $\tilde{O}(\min \{m, n^{3/2}\})$\\
Ghaffari and Kuhn \cite{ghaffari2018distributed} & Synchronous & $\tilde{O}(D + \sqrt{n})$ & $\tilde{O}(\min \{m, n^{3/2}\})$ \\
Present work & Asynchronous & $O(n)$ & $\tilde{O}(\min \{m, n^{3/2}\})$ \\ [1ex] 
\hline 
\end{tabular}
\label{table:nonlin} 
\end{table}
} Specifically, we show: 
\begin{theorem}
\label{nongeneral}
There exists an asynchronous algorithm in the KT1 CONGEST model that, w.h.p. 
computes the MST in $O(n)$ time and with $O(\min\{m, n^{3/2} \log^2 n\})$ messages.
\end{theorem}
This result achieves communication sublinear in $m$ when $m$ is sufficiently large, and is optimal for time when the diameter is $\Theta(n)$. We also prove the following more general theorem.

\begin{theorem}
\label{theogen}
Given an asynchronous MST algorithm with time $T(n, m)$ and message complexity of $M(n, m)$ in the KT1 CONGEST model, w.h.p., the MST in an asynchronous network can be constructed in $O(n^{1 - 2\epsilon} + T(n, n^{3/2 + \epsilon}))$ time and $\tilde{O}(n^{3/2 + \epsilon} + M(n, n^{3/2+\epsilon}))$ messages, for $\epsilon \in [0, 1/4]$.
\end{theorem}

For the BFS problem we show: 
\begin{theorem} \label{t:bfs}
In an asynchronous KT1 CONGEST model, a network with diameter $D$ can construct a spanning tree with $O(D + \sqrt{n})$ diameter using time $O(D^2 +n)$ and messages $\tilde{O}(n^{3/2}+nD)$.
\end{theorem}

\old{
\noindent
{\it Techniques:}
 A node is low-degree if its degree is $\tilde{O}(\sqrt{n})$, and high-degree otherwise. A set of $O(\sqrt{n})$ star nodes are self-selected randomly in the graph. We define a subgraph $G'$ as the subgraph on $G$ induced by all high-degree and star nodes. We first construct a spanning forest  $F$ on $G'$, where the sum of the diameter of all trees in this forest is $O(\sqrt{n})$. We do this in $O(n)$ time and with $\tilde{O}(n^{3/2})$ messages. Our new algorithm {\sc maximaltree} for this incorporates the asynchronous waiting technique of \cite{mashreghi_disc_2018}, to find the edges in $F$ and enable each high-degree and star node to learn its low-degree neighbors.

 As in \cite{ghaffari2018distributed} and \cite{gmyr2018}, we use the fact that the subgraph $G_{sparse}$ consisting of the edges in $F$  and all edges incident to at least one low-degree endpoint is a low diameter sparse spanner. 
To construct the nearly BFS tree, the network runs Gallager's Layered BFS algorithm \cite{gallager82} on $G_{sparse}$.\\ 
To construct the MST,  $F$ is used to find, in time linear in $n$, a minimum spanning forest $F_{min}$ of $G'$ and then the MST of the sparse subgraph consisting of the edges in $F_{min} $ and the edges incident to at least one low-degree node is found.  Applying Awerbuch's algorithm \cite{awerbuch1987optimal} to this subgraph results in an MST algorithm with time $O(n)$ and message complexity $\tilde{O}(n^{3/2})$.
}
\noindent
{\it Related work on MST:}
\textit{Asynchronous Model:} The first breakthrough for computing the MST was by Gallager, Humblet, and Spira \cite{gallager1983distributed}, who designed an asynchronous algorithm (known as GHS) which achieved $O(n \log n)$ time and $O(m + n \log n)$ message complexity. The algorithm was in the CONGEST model which allows for messages of size $O(\log n)$ bits. Later work gradually improved the time complexity of asynchronous MST computation to linear in $n$ \cite{chin1985almost, gafni1985improvements, awerbuch1987optimal, singh1995highly, faloutsos1995optimal, faloutsos2004linear}. The message complexity of the aforementioned papers is $O(m + n \log n)$. 

\textit{Synchronous Model:} Garay et al.\cite{garay1998sublinear} were the first to give a sublinear time $O(D + n^{0.614})$-round MST algorithm in the synchronous model. Kutten and Peleg \cite{kutten1995fast} gave an algorithm with  $O(D + \sqrt{n}\log^* n)$ time complexity, and Elkin \cite{elkin2004faster} provided an algorithm with $\tilde{O}(\mu (G, w) + \sqrt{n})$ round complexity. In his algorithm, $\mu (G, w)$ is the \emph{MST-radius} of the network. Although these algorithms can be simulated in an asynchronous network using a synchronizer \cite{awerbuch1985complexity, awerbuch1990network, awerbuch1993time, aggarwal1993time, kutten1997time, kutten1998asynchronous, burman2007time, awerbuch2007time}, either the superlinear time for initializing the synchronizer or their significant message overhead makes them unusable for our purposes.

For the KT0 model (also known as clean network) where nodes initially know only their own IDs, lower bounds of $\tilde{\Omega}(D + \sqrt{n})$ on the time complexity and lower bounds of $\Omega(m)$ on the message complexity have been proven \cite{elkin2006unconditional, sarma2012distributed, Peleg:2000:NLB:586844.586936, kutten2015complexity, awerbuch1987message}. There are algorithms that match both lower bounds simultaneously up to a polylogarithmic factor  \cite{pandurangan2016time, elkin2017simple}. 

In the KT1 model,  where nodes initially know the ID of their neighbors and in the presence of randomization, King et al. \cite{king2015construction}, provided an algorithm with $\tilde{O}(n)$ time and message complexity, which was the first algorithm that obtained sublinear message complexity in terms of the number of edges in the network. The time of their algorithm was later improved to linear in $n$ \cite{mashreghi2017time}. 
In 2018,  Ghaffari and Kuhn \cite{ghaffari2018distributed} and Gmyr and Pandurangan \cite{gmyr2018} examined time-message trade-offs in the synchronous CONGEST network.  Both these papers show how to build a low diameter sparse subgraph in the synchronous model. An $\Omega(m)$ lower bound for a restricted version of the synchronous KT1 model was shown in \cite{awerbuch1990trade}.

\noindent
{\it Organization:} Section \ref{s:model} defines the model,  Section \ref{s:algA} describes the algorithmic approach,  Section \ref{s:algs} gives the details of the algorithms, Section \ref{sublinearsection} discusses how our results can be generalized to obtain a trade-off between time and message complexity, and the Appendix gives a subroutine from \cite{ghaffari2018distributed}.
Finally, in Section \ref{conclusions}, we conclude the discussion.

\section{Model} \label{s:model}
\label{modelsec}
 We consider the asynchronous CONGEST model. In an asynchronous network, messages sent by nodes are delivered with arbitrary delays. Such communication is event-driven, where actions are taken upon receiving a message or waking up. We assume that all nodes wake up at the start. 
Time complexity in the asynchronous communication model is the worst-case execution time, if each message takes at most one time unit to deliver across one edge. The time for computations within a node is not considered.

 All nodes have knowledge of $n$, the size of the network, within a constant factor. In the CONGEST model, each message has $O(\log n)$ bits. ID's are unique and are taken from the range of $[1, poly(n)]$. We assume that messages to a receiver are numbered by a sender in the order in which they are sent to it, so that a node receiving the message can wait for the message from a sender with the sender's next number before acting.
W.l.o.g., we assume that the edge weights are unique and therefore, the MST is unique. If the edge weights are not unique they can be made unique while preserving the ordering of the weights, in the standard way (by multiplying the weight by $2^{2|ID|}$ and adding $x*2^{|ID|} + y $ where $x$ and $y$ are the uniquely labelled endpoints, $|ID|$ is the maximum length of an ID, and $x > y$). 

Nodes initially know their own ID and the ID of their neighbors. This is known as the KT1 model and is considered by some to be the standard model of distributed computing \cite{peleg2000distributed}. In weighted graphs,  initially, nodes only know the weight of their incident edges in the input graph G. We assume that all nodes wake up at the same time. In a distributed network problem for constructing a subgraph like MST,  the objective is that, upon termination of the protocol, all nodes must know which of their incident edges belong to the subgraph.  We explain the algorithm assuming that the input graph is connected. However, the protocol directly applies to disconnected graphs as well.  \\[-.35in]

\section{Algorithmic approach} \label{s:algA}
A common approach in distributed protocols for computing the MST is using {\it Boruvka's algorithm}. (See \cite{nevsetvril2001otakar}.)
Boruvka's Algorithm runs in $O(\log n)$ phases. The idea is to maintain a subgraph of the minimum spanning forest during the algorithm. Initially, each node is a tree in the spanning forest. In each phase, each tree (also called a fragment of the MST) computes a minimum weight edge leaving the fragment. Then, fragments are merged using these minimum outgoing edges. Each fragment is rooted at a specific node called the leader. When some fragment B finds a minimum outgoing edge to a fragment A rooted at node $a$, it requests to merge with A. If A accepts the merge, B becomes a subtree of A, and the two fragments A and B, and the edge used for merging them become one fragment rooted at node $a$. The implementation of the algorithm in the synchronous model allows a constant fraction of fragments to merge together in each phase. Therefore, $O(\log n)$ phases suffice for all fragments to merge into one fragment, which will be the final MST. 

We assume that each fragment of the MST has an \textit{identity} which we refer to as the fragment ID. Fragment ID is the node ID of the fragment’s leader and can change for any fragment over the course of the algorithm. Every node in a fragment is aware of the fragment ID.

It is difficult to mimic the parallelism of the Boruvka algorithm in the asynchronous model. The GHS algorithm succeeds in doing so by maintaining a rank for each fragment. Fragments use their minimum weight outgoing edge to merge into fragments of equal or higher ranks. In this way a fragment doubles in size with each full search for a minimum weight outgoing edge. The GHS algorithm avoids incurring a later full search cost when the minimum weight outgoing edge is to a fragment of lower rank.  Without communicating across each edge, it seems impossible to simulate this aspect of the GHS algorithm.  Mashreghi and King \cite{mashreghi_disc_2018} drops the attempt to run a parallel Boruvka-style algorithm and instead grows its MST from a single source node. In doing so, it runs in worst case time proportional to the total number of messages. 

Here we accomplish some sort of parallelism by starting with ``initial" fragments of height 1, formed by high degree nodes and star nodes in parallel. A node is  {\it high-degree} if its degree is at least $\sqrt{n} \log^2 n$. Otherwise, it is a {\it low-degree} node.
A node selects itself to be a {\it star node} with probability of $\frac{c}{\sqrt{n} \log n}$ where $c$ is a constant dependent on $c'$ so that each high-degree node is adjacent to a star node with probability $1-1/n^{c'}$. 
 Let $G'$ be the subgraph on $G$ induced by all high-degree and star nodes.  We can construct a spanning forest  $F$ on $G'$, from the initial fragments by adding $O(\sqrt{n}/\log n)$ edges such that the sum of the diameter of all trees in this forest is $O(\sqrt{n}/\log n)$. Thus we can afford to add these edges sequentially in the worst case, spending a time per edge proportional to $\log^2 n*(\hbox{diameter of a maximum tree})$, for a total of $O(n)$ time. 

To find minimum outgoing edges, one approach is to have nodes test all of their incident edges, in the order of weight, to see whether an edge is outgoing. This results in $\Omega(m)$ messages. King et al. \cite{king2015construction} provided an \emph{asynchronous} subroutine, called FindAny, that with constant probability,  finds an edge leaving a fragment $T$ and uses only $\tilde{O}(|T|)$ messages, where $|T|$ is the number of nodes in $T$. This subroutine is used in \cite{ghaffari2018distributed} and \cite{mashreghi_disc_2018} to compute the MST. The following lemma on FindAny subroutine will be used throughout the algorithm.  $FindAny(E')$ means that FindAny is performing this search on a subset $E'$ of the edges. 

\begin{lemma}{\cite{king2015construction}}
\label{findAnyLemma}
With probability 1/16, FindAny succeeds and finds an outgoing edge from fragment $T$. Otherwise, it fails and returns $\emptyset$. If it succeeds, the returned edge is chosen uniformly at random from the set of all outgoing edges. If there is no outgoing edge, FindAny returns $\emptyset$ with probability 1. In the asynchronous model, this takes time proportional to the height of $T$ and $O(|T|)$ messages.
\end{lemma}
To find the minimum outgoing edge, it is possible to use FindAny and do a binary search on the weight of the desired edge. To do this, assuming that all edge weights are in range $[l, h)$, nodes first consider only those incident edges whose weight is in $[l, (l+h)/2)$, and run FindAny $O(\log n)$ times. Then with high probability, if there is an outgoing edge in this range it will be found and we can narrow down the search to lighter edges. Otherwise, it must be that the minimum outgoing edge is in $[(l+h)/2, h)$. In each step of the binary search, the current bounds are broadcast to all nodes in the tree. After $O(\log n)$ steps of the binary search the minimum outgoing edge is found. This takes $O(n \log^2 n)$ messages, and $O(height(T) \log^2 n)$ time overall, ($height(T)$ is the height of $T$.)

Two other subroutines used in this paper are from \cite{mashreghi_disc_2018}: \\
\noindent
$ApproxCut(T)$: returns an estimate in $[k/32, k]$ where $k$ is the number of outgoing edges from $T$ and $k > c \log n$ for $c$ a constant.
It requires $O(n \log n)$ messages.\\
\noindent
$ThresholdDetection(k)$: The leader is informed w.h.p. when the number of events experienced by the nodes in its tree reaches the threshold $k/4$. The event here is the receipt of $\langle \textit{Low-degree} \rangle$ over an outgoing edge. 

\subsection{Outline of algorithms}
\old{A node is  {\it high-degree} if its degree is at least $\sqrt{n} \log^2 n$. Otherwise, it is a {\it low-degree} node.
A node selects itself to be a {\it star node} with probability of $\frac{c}{\sqrt{n} \log n}$.

Our algorithm computes the MST with probability $1 - 1/n^{c'}$, where $c'$ is any constant, and $c$ in self-selecting star nodes is a constant depending on $c'$ so that each high-degree node is adjacent to a star node w.h.p.. An edge belongs to the subgraph $G'$ if and only if both of its endpoints are either high-degree or star nodes. Note that $G'$ is not necessarily connected. }

 Recall that an edge belongs to the subgraph $G'$ if and only if both of its endpoints are either high-degree or star nodes.
 
\begin{enumerate}

\item Initial fragments are formed in parallel consisting of star nodes and their high degree neighbors. Our new algorithm {\sc maximaltree}  incorporates the asynchronous waiting technique of \cite{mashreghi_disc_2018} to find the edges in $F$ and enable each high-degree node to learn its low-degree neighbors while building a spanning forest $F$ on $G'$ from these initial fragments. See Section \ref{s:maximaltree}. 

To construct the nearly BFS tree, the network runs Gallager's Layered BFS algorithm \cite{gallager82} on $G_{sparse}$, the subgraph of edges in $F$ and all edges incident to at least one low-degree node.  (Section \ref{s:bfs})

\item We compute the minimum weight spanning forest $F_{min}$ on $G'$ in Section \ref{s:fmin}. We use an idea of \cite{ghaffari2018distributed}, along with the fact that the obtained trees in part (1) all have diameter of $O(\sqrt{n})$. This is done using the low-diameter trees  in $F$, to simulate the MST algorithm of \cite{ghaffari2018distributed} in the asynchronous model on each connected component of $G'$.

\item  In Section \ref{s:mst}, we define $S_{min}$ to be the edges in $F_{min}$ and the edges incident to at least one low-degree node. We run an asynchronous MST algorithm with $O(n)$ time and $O(m)$ message complexity  (e.g. \cite{awerbuch1987optimal}) on $S_{min}$. The result is the MST of $G$.

\end{enumerate}
The challenge in part (1) is to modify the algorithm of \cite{mashreghi_disc_2018}, to have star nodes grow fragments of the spanning forest only on $G'$ (not including any low-degree node). In part (2), we observe that the algorithm of \cite{ghaffari2018distributed} consists of a number of Boruvka style phases where the computations in each phase are inherently asynchronous. We show that using the low-diameter trees from part (1), we need to synchronize the nodes only at the beginning of each such phase. This allows us to run, without any asymptotic overhead on the complexity, the MST algorithm of \cite{ghaffari2018distributed} in the asynchronous network. We run the algorithm on each connected component of $G'$. This results in the minimum spanning forest for part (2) of our algorithm. 

For the last part, the challenge is that there is no global clock to announce the beginning of these three parts to nodes. Therefore, we have all nodes start running the MST protocol of part (3) while part (1) and (2) are being computed. To do this, the MST protocol is delayed by the high-degree or star nodes if they have not yet computed their corresponding tree in $F_{min}$. We show, however, that this approach in coordinating the protocols does not affect the asymptotic time and message complexity of the MST protocol. 
\\[-0.35in]

\section{Algorithms}  \label{s:algs}  
\subsection{{\sc maximaltree:} Constructing a Spanning Forest $F$ on $G'$} \label{s:maximaltree}
Initially, similar to the idea of \cite{ghaffari2018distributed}, we form a number of height-one fragments around the star nodes. Star nodes send a message to their neighbors and await their response. We observe that w.h.p. each high-degree node is adjacent to a star node and assume this property holds during the course of our analysis. Each high-degree node considers itself to be the child of the first star node it has heard from and responds accordingly. This gives us the \emph{initial fragments} of the spanning forest, where each fragment is formed by a star node and a subset of its neighbors. Since a star node is self-selected with probability of $\frac{c}{\sqrt{n} \log n}$, w.h.p., there are at most $O(\sqrt{n}/\log n)$ star nodes.

\old{
\subsection{Implementation of the Spanning Forest Algorithm} We now describe the details of the spanning forest construction.
\subsubsection{A - Initialization and preliminaries} All nodes first run the Initialization procedure (Algorithm 1). Star nodes self-select themselves with probability of $\frac{1}{\sqrt{n} \log n}$ and then send a $\langle Star \rangle$ message to all of their neighbors. Each recipient of  $\langle Star \rangle$  becomes a child of a Star node iff the recipient is of high degree, is not a star, and  the $\langle Star \rangle$ message is the first such message it has received. In any case, a recipient responds to let the sender know whether or not it has become its child.  The initial height-one fragments are formed in this way by stars and their children. 

Upon waking up, low-degree nodes that are not stars send a $\langle \textit{Low-degree} \rangle$ message to all of their neighbors. 

The following information is stored: 

\noindent
 {\it  Fragment IDs:} Each fragment has an ID equal to that of the star node ``leading" the fragment. Each high-degree and star node $x$ has a variable $xID$ set to the ID of the fragment it most recently joined. We use $xID$ to refer both to the fragment ID and the fragment tree of node $x$. f after initialization all high-degree nodes belong to a fragment. 
 
\noindent
{\it Local information of a high-degree node:} Whenever a high-degree node receives a $\langle \textit{Low-degree} \rangle$ message, it remembers that incident edge so it can later exclude it from the search for outgoing edges. 
}

\noindent {\it From initial trees to $F$:} Now, we only need to connect these initial fragments using $O(\sqrt{n}/\log n)$ edges. The challenge is to find an outgoing edge to a high degree node or a star node. Each phase of the MaximalTree protocol returns such an edge if there is one, or until there is a ``terminal" run when all nodes in the star's fragment have received messages from all their (low-degree) neighbors outside the fragment and none of these neighbors are high degree or star nodes.   Each leader of a fragment runs MaximalTree until it no longer is a leader of its fragment or a terminal run is reached.

\noindent
{\it Merging fragments:} After finding an outgoing edge $e$ to a high degree or star node, a fragment (say $A$)  sends a merge request message along the edge with its fragment ID to the other fragment (say B) and waits for a response. If  $ID(A) < ID(B)$ then the endpoint in $B$ accepts the request and  fragment $A$ then updates the fragment ID of all of its nodes to $B$, that is $B$'s leader becomes the leader of the combined fragment. However, if $A$'s ID is greater, the endpoint of $e$ in B waits (and $A$ waits) until:\\
i)  $B$ selects $e$ or another edge with an endpoint in $A$, in which case $B$ merges with $A$ along the edge selected by $B$ and $B$'s ID become equal to $A$'s.  The leader of $A$ becomes the leader of the merged fragment and receives a message that its attempted merge is rejected.  \\
ii) $B$ selects an outgoing edge to a different fragment $A'$ and $ID(B)$ is updated so that $ID(B) > ID(A)$. $A$ then updates its ID to $B$'s as described above.  \\[-0.30in]

\begin{observation} \label{o:height}
The sum of the diameters of the trees in the spanning forest F is $O(\sqrt{n}/\log n)$.
\end{observation}
\begin{proof}  
There are $O(\sqrt{n}/\log n)$ initial fragments. Since each fragment consists of a star node and its high degree neighbors, the height of each fragment is one. Therefore, the height of any tree in F, obtained by adding edges between the initial fragments, is no more than twice the final number of star nodes in that tree which is $O(\sqrt{n}/ \log n)$. Since diameter of a tree is at most twice the height, the sum of the diameters of the trees in F is $O(\sqrt{n}/\log n)$. 
\end{proof}

\old{
\begin{observation}
\label{l:awareness}
W.h.p., every high degree node knows all its (low-degree) neighbors which lie outside its tree $T$ in $F$ once $\sc{maximaltree} $ is terminated by $T$'s leader. 
\end{observation}

}

\subsubsection{Analysis of \sc{maximaltree}}
\begin{lemma}
\label{lemmak1}
Computing the maximal trees of the spanning forest F requires $O(n)$ time and $O(n^{3/2} \log^2 n)$ messages w.h.p. \\[-0.35in]
\end{lemma}
\begin{proof}
In the initialization, there are at most $n$ low-degree non-star nodes that send to at most $O(\sqrt{n} \log^2 n)$ of their neighbors. Also, there are at most $O(\sqrt{n}/\log n)$ star nodes that send to all of their neighbors and receive a response. This overall takes $O(n^{3/2} \log^2 n)$ messages and $O(1)$ time.

To find an outgoing edge, in each phase, each fragment $T$ runs FindAny $O(\log n)$ times. Each phase uses  $O(|T| \log n)$ messages and $O(height(T) \log n)$ time. Moreover, finding an outgoing edge to a high-degree or star node may require $O(\log n)$ phases of sampling and waiting. Therefore, each fragment requires $O(|T| \log^2 n)$ messages and $O(height(T) \log^2 n)$ time to find an outgoing edge. 

A merge request is rejected only when an outgoing edge becomes an internal edge because of another merge, and hence the number of rejected merge requests is $O(\hbox{number of merges})$. There are at most $O(\sqrt{n}/\log n)$ merges that have to be performed using the outgoing edges on the initial height-one fragments. Also, height of a fragment is always $O(\sqrt{n} /\log n)$ (Observation \ref{o:height}). Therefore, the overall time will be $O((\sqrt{n} /\log n)^2 \cdot \log^2 n) = O(n)$. The overall message complexity is $O(n \cdot \log^2 n \cdot \sqrt{n} /\log n) = O(n^{3/2} \log n)$.
\end{proof}

We also prove the following lemma regarding the correctness of the algorithm.
\begin{lemma}
\label{lemmak2}
With high probability, MaximalTree protocol always makes progress. When it terminates, the obtained trees are maximal trees in F and all outgoing edges are explored. 
\end{lemma}
\begin{proof}
When the fragments find their outgoing edges, the policy for merging always allows the fragment with lowest ID to be merged with another fragment. So, in terms of merging, the algorithm always makes progress. We should only argue that if an outgoing edge to a high-degree or star node exists, it will be found w.h.p.. We call such an edge a valid outgoing edge for the sake of this proof.\par 
If there is a valid outgoing in the sample, we are done. Otherwise, the algorithm explores all low-degree non-star nodes outside the fragment by waiting for them to send their $\langle \textit{Low-degree} \rangle$ messages. This happens w.h.p. since we use ApproxCut and ThresholdDetection. ApproxCut approximates the cut within a constant factor and ThresholdDetection only signals the leader when a constant fraction of the $\langle \textit{Low-degree} \rangle$ messages are received. Therefore,  each time these two subroutines are applied, a constant fraction of the messages that we expect to be received, will be received with high probability. So, after repeating this $O(\log n)$ times, the number of $\langle \textit{Low-degree} \rangle$ messages that have not been received becomes so low that they will not interfere with finding a valid outgoing edge. In the worst-case, when the majority of such messages are received and their corresponding edges are excluded from the search, FindAny finds a valid outgoing edge.\par
If there is no valid outgoing edge, the fragment tree is maximal. The last set of sampled edges before termination includes all remaining outgoing edges,  and the algorithm terminates.
\end{proof}

\begin{algorithm}
\caption{Initialization of the minimum spanning tree algorithm. Every node runs this protocol independently upon wake up.}
\begin{algorithmic}[1]
\Procedure{Initialization}{}
\label{AinitAlg}
\State Every node selects itself to be a \emph{star} node with probability of $\frac{c}{\log n \sqrt{n}}$.
\State Star nodes send a $\langle Star \rangle$ message to all of their neighbors and wait for the response of each message. In response, if a node is high-degree and this is the first star it has hear from it sends back $\langle \textit{Child} \rangle$. Otherwise, the node responds $\langle \textit{Not-Child} \rangle$. If the node receiving the $\langle Star \rangle$ message is also a star, it responds by $\langle \textit{Star-node}\rangle$.
\label{child}
\State Low-degree nodes that are not a star send $\langle \textit{Low-degree} \rangle$ messages to all of their neighbors.
\EndProcedure
\end{algorithmic}
\end{algorithm}

\begin{algorithm}
\caption{Finds a maximal tree of the spanning forest $F$ in $G'$. $x$ is any star node that is also the leader of this fragment.}
\begin{algorithmic}[1]
\Procedure{MaximalTree}{$x$}
\label{maximal_tree}
\Repeat {\textit{ } \texttt{//beginning of a phase}}
\State $counter \leftarrow 0, A \leftarrow \emptyset$.
\label{firstlineafter}
\While{$counter <  16 c \log n$}
\label{sampleloop}
\State  Leader calls $FindAny(E \setminus A)$, where nodes exclude from the search the edges that they have received a $\langle \textit{Low-degree} \rangle$ message from, i.e., nodes that are low-degree but not star.
\If {$FindAny$ is successful and finds an edge $(u, v)$ ($u \in T$ and $v \notin T$)}
\State $A = A \cup {(u,v)}$. \texttt{//the set of outgoing edges}
\State $u$ sends a message and asks for $v$'s degree, and whether or not it is a star node. $u$ then sends up the result to the leader.
\label{putinfound}
\EndIf
\State $counter \leftarrow counter + 1$.
\EndWhile
\If {$\exists (u, v) \in A$ s.t. $v$ is high-degree or star}
\State Leader chooses an edge to a high-degree or star node arbitrarily, and sends a $\langle Merge, ID \rangle$ message over it. ($ID$ is this fragment's ID.)
\State If accepted, this fragment is updated to be a subtree of fragment $vID$ rooted at node $v$.
\State If rejected, the leader starts a new phase.
\ElsIf {$|A| < 2 \log n $}
\State Leader terminates the algorithm.
\Else { \texttt{ // waiting} }
\State $r \leftarrow ApproxCut() / 2$. Then leader calls $ThresholdDetection(r)$.
\State Leader waits to receive a trigger message and then starts a new phase.
\label{waitperiod}
\EndIf
\Until
\EndProcedure
\end{algorithmic}
\end{algorithm}

\begin{algorithm}
\begin{algorithmic}[1]
\Procedure{ReceiptOfMerge}{$\langle Merge, tID\rangle $}
\State When node $x$ receives the message $\langle Merge, tID \rangle $ from node $t$:
\State If $x$ is a high-degree node it waits to hear from at least one star node. Else, if $x$ is a star node it waits to hear the response of its initialization messages. 
\If{$tID < xID$}
\State $x$ immediately responds by $\langle Accept, xID \rangle$, and considers $t$ as a child from now on.
\ElsIf{$tID > xID$}
\State $x$ delays the response until $xID$ becomes greater than or equal to ($\geq$) $tID$.
\Else {\textit{ } $tID = xID$}
\State $x$ rejects the merge by a $\langle Reject \rangle$ message.
\EndIf
\EndProcedure
\end{algorithmic}
\end{algorithm}

\subsection{ Constructing a spanning tree of height $O(D + \sqrt{n})$} \label{s:bfs}
The Layered BFS algorithm for asynchronous networks due to Gallager \cite{gallager82} (see James Aspnes's online notes \cite{aspnesnotes}, Section 4, pp 25)  runs in time $O(D^2)$ and uses messages $O(E+VD)$ on any graph with diameter $D$. This simple algorithm assumes an initiator which starts the algorithm, and then the tree is built one layer at a time, reporting back to the root when the layer is done. 

Let  $G_{sparse}$ be the subgraph of $G$ consisting of edges in $F$ and the edges with at least one low-degree endpoint. 
From Observation \ref{o:height}, and as observed in \cite{ghaffari2018distributed}, page 5-6, we know $G_{sparse}$ has diameter $O(D_G+ \sqrt{n})$ where $D_G$ is the diameter of $G$.
From Lemma \ref{lemmak2}, each node knows which of its neighbors are in $G_{sparse}$. 

The full algorithm is as follows: At the start, all nodes in $G$ awake to construct $F$. An initiator, once it has finished with the construction of $F$, initiates the Layered BFS algorithm on $G_{sparse}$. Each node responds to the messages sent by the Layered BFS algorithm only after it has completed the construction of $F$. For this algorithm, each node considers its incident edges to be only those in $ G_{sparse}$. We obtain the following: 
 BFS in $G_{sparse}$ is constructed of diameter $O(D_G + \sqrt{n})$ using time $O(D_G^2 +n)$ and messages $\tilde{O}(n^{3/2}+nD_G)$, proving Theorem \ref{t:bfs}.  \\[-.35in]

\subsection{Constructing the Minimum Spanning Forest $F_{min}$ on $G'$} \label{s:fmin}
It is easier to describe the algorithm for constructing $F_{min}$ as an MST algorithm on each connected component of $G'$. Formally, let us fix a connected component $C$ in $G'$. Let $T$ be the low-diameter tree computed on $C$ from part (1). Now, we describe how to compute $T_{min}$, the minimum weight spanning tree of $C$, using $T$ whose diameter is bounded by $O(\sqrt{n})$ (we omit the $\log n$ division as it is not needed in this part). \par 
We simulate the synchronous MST algorithm of \cite{ghaffari2018distributed} in the asynchronous model, on $C$. So, our algorithm in this part closely follows their algorithm. (Please see the appendix for a detailed description of their algorithm since we omit the details that do not affect the synchronization process.) Our simulation is in a way that results in no asymptotic overhead on the time and message complexity of \cite{ghaffari2018distributed}. Our algorithm has two parts: \par 

\noindent
\textit{(A) Computing fragments of $F_{min}$ with low diameter:} Let $r$ be the root of $T$. This part consists of $O(\log n)$ Boruvka phases. However, in each phase only fragments with height less than $\sqrt{n}$ can look for minimum outgoing edges. Also, only certain merge requests will be accepted. The objective is to grow $O(\sqrt{n})$ MST fragments where the diameter of each fragment is bounded by $O(\sqrt{n})$ (in particular between $\sqrt{n}$ and $5\sqrt{n}$). The steps are as follows:
\begin{enumerate}
\item Initially, each node is a fragment.
\item{Begin the search:} $r$ tells all nodes (via a broadcast) to begin the search for the minimum outgoing edges. Then, each fragment computes its minimum outgoing edge using a binary search and FindAny (see Section \ref{s:algA}). When nodes in a fragment know that the minimum outgoing edge is computed, they let $r$ know using a convergecast.  Note that they just let $r$ know that they have finished the computation and do not send the found edge. $r$ waits until all fragments compute their minimum outgoing edges.
\item {Begin the merge:} $r$ tells all nodes to merge using the recently found minimum outgoing edges. All fragments then send merge requests over their minimum outgoing edges. If the merge is accepted, the fragment IDs of the nodes will be updated. Once fragments finished this part, they let $r$ know by a convergecast.
\item {Begin to truncate:} $r$ tells all fragments to make sure that their height stays in $[\sqrt{n}, 5 \sqrt{n}]$. Once the fragments merge, it is possible that the height of the resulting fragments exceeds $5 \sqrt{n}$. However, because of the merging rules in \cite{ghaffari2018distributed}, it is guaranteed the height will not exceed $5 \sqrt{n}$ by more than a constant factor. To make sure that all fragments have height in $[\sqrt{n}, 5 \sqrt{n}]$ we do as follows. Each fragment leader broadcasts a message to all nodes so that all nodes know their distance from the leader. Whenever the distance of a node from the fragment leader becomes a multiple of $\sqrt{n}+1$, the edge between that node to its parent is discarded. In case that the height of the remaining subtree gets below $\sqrt{n}$ that operation is undone. Once truncating is done, nodes notify $r$ via a convergecast.
\item {Check the phase:} $r$ computes the minimum height of all fragments via broadcast and convergecast. If the minimum height is $< \sqrt{n}$, $r$ goes to Step 2 and announces the beginning of a new phase. Otherwise, $r$ announces the beginning of part (B).
\end{enumerate}
We call the fragments obtained in part (A) \emph{old fragments}. Also, we call the leader of these fragments \emph{old leaders}. 
Once $r$ announces the beginning of part (B), no more synchronization is required as the computations in part (B) are inherently asynchronous.

\noindent
\textit{(B) Completing $F_{min}$:} Since the diameter of $F_{min}$ can be as large as $\Theta(n)$, in this part, nodes communicate through a different network. In fact, nodes communicate via the old fragments and through the spanning tree $T$. The key point in this part is to grow fragments in a way that each fragment  consists of a number of old fragments. \par
To find the remaining $O(\sqrt{n})$ minimum outgoing edges, and merge the old fragments, fragments rely on $r$ to do the computation. In particular, nodes first send the necessary information to their old leaders. Then, the old leaders, send their information (via $T$) to $r$. So, each message travels $O(\sqrt{n})$ hops. \par 
Then, $r$ computes the minimum outgoing edges and sends the instructions regarding the merges, back to the nodes. Once the minimum outgoing edges are computed and $r$ knows how the fragment IDs should be updated, information regarding the updates can be passed down to all nodes by reversing the direction of the convergecast messages.\par \textbf{Note:} Since each fragment in part (B) consists of of a number of old fragments, $r$ should pass the ID updates only to the old leaders. Then, the old leaders can pass the updates to all nodes in the old fragment.

\begin{lemma}
\label{lemmak3}
Computing $F_{min}$ requires $\tilde{O}(n)$ messages and $\tilde{O}(\sqrt{n})$ time.
\end{lemma}
\begin{proof}
In part (A) there are $O(\log n)$ phases since the fragments are synchronized via $T$ and in each phase a constant fraction of all fragments can be merged. In each phase, fragments use a binary search and FindAny to find the minimum outgoing edges which takes $\tilde{O}(n)$ messages and $\tilde{O}(\sqrt{n})$ time. Merging and truncating in each phase needs a constant number of  broadcasts and convergecasts which take $O(n)$ messages and $O(\sqrt{n})$ time. Therefore, overall, part (A) takes  $\tilde{O}(n)$ messages and $\tilde{O}(\sqrt{n})$ time. \par 
In part (B), we only broadcast the beginning of part (B). Afterwards, all computations are performed exactly in the same way as \cite{ghaffari2018distributed}. Part (B) has also $O(\log n)$ phases. However, we do not need to synchronize phase by phase. In each phase, all nodes give the necessary information to the old leaders. This takes $\tilde{O}(n)$ messages and $\tilde{O}(\sqrt{n})$ time. Then, the old leaders give the necessary information to $r$. This also takes $\tilde{O}(n)$ messages and $\tilde{O}(\sqrt{n})$ time since there are $O(\sqrt{n})$ messages that have to be pipelined and travel $O(\sqrt{n})$ hops.

Over all connected components of $G'$, the algorithm takes $\tilde{O}(n)$ messages and $\tilde{O}(\sqrt{n})$ time, and the lemma follows.
\end{proof}

The correctness of part (2) follows from the correctness of the synchronous algorithm of \cite{ghaffari2018distributed}. The computations in each of the steps in part (A) and the whole part (B) can be simulated asynchronously since they do not depend on a global clock to be correct. Therefore, such simulation does not affect the correctness of the protocol.  \\[-.35in]

\subsection{Constructing the MST from $F_{min}$:} \label{s:mst}
We now have constructed the desired spanning subgraph $S_{min}$ on $G$. This subgraph consists of all edges that have at least one low-degree endpoint, and the edges of the minimum spanning forest $F_{min}$. Therefore, $S_{min}$ has $\tilde{O}(n^{3/2})$ edges. \par 
To construct the final MST we can use any $O(n)$ time asynchronous MST algorithm that requires no more than $O(m)$ messages on dense graphs. We use the algorithm of \cite{awerbuch1987optimal}. 

However, the only catch is that in an asynchronous network we cannot first complete parts (1) and (2) of the algorithm and then move to part (3). To solve this, we have all low-degree nodes run the MST protocol, right after they sent their initialization messages. High-degree and star nodes only begin to participate in the MST protocol once they have computed their corresponding tree in $F_{min}$. Computing $F_{min}$ requires $O(n)$ time (for part (1) and (2)). Therefore, the overall delays that high-degree and star nodes can cause is no more than  $O(n)$, which is the same as the time complexity of the MST protocol. We now prove the following lemma which implies that running an MST algorithm on $S_{min}$ computes the correct MST of $G$.

\begin{lemma}
\label{lemmak4}
All edges of the final MST must be in subgraph $S_{min}$.
\end{lemma}
\begin{proof}
Suppose on the contrary that there is some edge in the MST that is not in the set of edges of the subgraph $S_{min}$. This edge cannot have a low-degree endpoint because all such edges belong to the subgraph. Therefore, it must be an edge that is in the edges of $G'$ but does not appear in $F_{min}$. However, adding this edge to the corresponding tree in $F_{min}$ creates a cycle. By removing the maximum weight edge on that cycle we obtain a lighter tree which contradicts the fact that $F_{min}$ is the minimum spanning forest of $G'$.
\end{proof}

Correctness of this part follows from the correctness of the MST protocol that is used. The only thing we do in this part is that we let high-degree and star nodes participate in the protocol only when they finished computing their trees in $F_{min}$. This will only cause an additive delay on the MST protocol and will not affect the correctness. 

Theorem \ref{nongeneral} follows from Lemmas \ref{lemmak1}, \ref{lemmak2}, \ref{lemmak3}, and \ref{lemmak4}. \\[-.35in] 
\section{Trade-offs}
\label{sublinearsection}
The only parts in our algorithm that require linear time in $n$ are part (1) and part (3). Part (3) is linear due to the asynchronous MST algorithm that we use on $S_{min}$. However, in part (1), the time complexity comes directly from the number of star nodes. If we have $n^{1/2 - \epsilon}$ star nodes, the time complexity of this part is $\tilde{O}(n^{1 - 2 \epsilon})$. Since part (2) has a time complexity of $\tilde{O}(\sqrt{n})$, we only consider $\epsilon \in [0, 1/4]$. However, we have to increase the threshold for low-degree nodes to $n^{1/2 + \epsilon}$, as well. This will result in a message complexity of $\tilde{O}(n^{3/2 + \epsilon})$ which is interesting if the input graph has asymptotically more edges. For instance, by selecting $n^{1/4}$ ($\epsilon = 1/4$) star nodes, we have $\tilde{O}(n^{1/2})$ time and $\tilde{O}(n^{7/4})$ message complexity. \par
Notice that we presented our algorithm for $\epsilon = 0$; however, we optimized the $\log n$ factor to make sure that the time complexity remains $O(n)$.  Overall, the time and message complexity of part (1) and (2) of our algorithm will be $\tilde{O}(n^{1 - 2 \epsilon})$ and $\tilde{O}(n^{3/2 + \epsilon})$ respectively. And the subgraph $S_{min}$ will have $\tilde{O}(n^{3/2 + \epsilon})$ edges. Therefore, by applying 
 an asynchronous MST algorithm with time $T(n, m)$ and message complexity of $M(n, m)$, we get an algorithm with time $O(n^{1 - 2\epsilon} + T(n, n^{3/2 + \epsilon}))$ and message complexity of $\tilde{O}(n^{3/2 + \epsilon} + M(n, n^{3/2+\epsilon}))$. This proves Theorem \ref{theogen}.
\section{Conclusions}
\label{conclusions}
The most important question is whether it is possible to find an MST algorithm in the asynchronous model, that requires time sublinear in $n$ and messages sublinear in $m$. We know from the previous section that this can be done if there exists an asynchronous MST algorithm that takes sublinear time if the diameter of the network is low, and has $\tilde{O}(m)$ message complexity. To the best of our knowledge no such algorithm exists. Similarly, is it possible to have an asynchronous breadth-first-search (BFS) algorithm that requires $\tilde{O}(D)$ time and $\tilde{O}(m)$ messages? Having such algorithm would help in synchronizing the network and would allow us to achieve sublinear time by following the same strategy as part (2) of our algorithm on the whole graph.




\newpage

\bibliography{mybib}

\appendix

\section{MST algorithm of \cite{ghaffari2018distributed}} In this section, we describe the algorithm of \cite{ghaffari2018distributed} for computing the MST of a connected component if a spanning tree of low diameter already exists in that component. We call this low diameter spanning tree $T$ and assume that $r$ is the root of $T$. In \cite{ghaffari2018distributed}, this algorithm is performed on the whole graph when a spanning tree of diameter $\tilde{O}(\sqrt{n} + D)$ was computed on $G$ using a breadth-first-search algorithm. The two parts of the algorithm are as follows: \par 

\textit{Part A - Computing MST fragments with low diameter:} An important point is that in this part merging is done differently from the standard Boruvka algorithm. In each phase, before fragments start merging, each fragment flips a fair coin. Then, if a fragment $A$ sends a merge request to a fragment $B$, the merge is only accepted if $A$ has a Tail coin, and $B$ has a Head coin. \par
Moreover, not all fragments look for minimum outgoing edges. In fact, fragments whose diameter is less than $\sqrt{n}$ are \emph{active}, and the other fragments are \emph{inactive}. Only active fragments can look for the minimum outgoing edge. However, \emph{inactive} fragments may still accept the merge request of other fragments, given that the active fragment has a Tail coin and the inactive fragment has a Head coin. \par 
There are $O(\log n)$ phases. In each phase, active fragments look for the minimum outgoing edge and send merge requests along those edges. If accepted, the fragments are updated. The objective is to keep the height of all fragments (active or inactive) in range $[\sqrt{n}, 5 \sqrt{n}]$. After the merging is done it is possible that the height of some fragments exceeds $5 \sqrt{n}$. However, because of the coin flip that the fragments use, the merges always happen at the center of a Head fragment. Therefore, after merging the height will be no more than $\sqrt{n} + 5 \sqrt{n} + 1 = O(\sqrt{n})$. \par 
To truncate the trees, fragment leaders do a broadcast to let all nodes know their distance from the leader. Whenever distance of a node $v$ becomes a multiple of $\sqrt{n} + 1$ the edge from $v$ to its parent is discarded. However, it is possible now that the subtree rooted at $v$ has a height of less than $\sqrt{n}$. To detect this, leaves do a convergecast so that each node on higher levels knows the height of the subtree. If $v$ is a node whose subtree has a height  less than $\sqrt{n}$ it will  undo the truncate operation and puts the edge back in the tree.\par 
Part (A) terminates when there are no active fragments left. Fragments computed in this part are called \emph{old fragments}.

\textit{Part B - MST growth beyond low-diameter fragments:} The diameter of the final MST could be as large as $n$. To make sure that this part is performed in sublinear time in $n$, MST fragments rely on $r$ to compute their minimum outgoing edges for them, and help them merge. The objective is to find the $\tilde{O}(\sqrt{n})$ minimum outgoing edges and link the old fragments. The key point here is that each MST fragment is composed of a number of old fragments. This allows, $r$ to send the instructions to only the old leaders. Then, the old leaders pass them down to the nodes in the old fragment. \par 
In FindAny, all a node has to do is to compute linear sketches from its incident edges that it wants to consider in the search. These sketches have a size of $\tilde{O}(1)$. To find the minimum outgoing edges, the fragments do as follows. Nodes compute their sketches and give them to their old leader via a convergecast. Then, the old leaders give all of the sketches along with their current fragment IDs to $r$. Then, $r$ can compute the minimum outgoing edges by performing the binary search. Once $r$ computes the minimum outgoing edge, it broadcasts the corresponding minimum outgoing edges back to the old leaders, and the old leaders pass it down to all nodes in their old fragment. Similarly, $r$ tells the nodes how they should update their fragment ID.

\end{document}